\documentclass[reprint,amsmath,amssymb,aps,pra,onecolumn]{revtex4-2}

\usepackage{dcolumn}
\usepackage{bm}

\usepackage{relsize}
\usepackage[utf8]{inputenc}
\usepackage[english]{babel}
\usepackage[T1]{fontenc}
\usepackage{hyperref}
\usepackage{graphicx}
\usepackage{xcolor}

\usepackage{amsthm}
\usepackage{amsmath} 
\usepackage{amssymb}
\usepackage[colorinlistoftodos]{todonotes}
\usepackage{ulem}
\usepackage{blkarray}
 
\newtheorem{theorem}{Theorem}
\newtheorem{corollary}[theorem]{Corollary}
\newtheorem{lemma}[theorem]{Lemma}
\newtheorem{definition}[theorem]{Definition}
\newtheorem{remark}[theorem]{Remark}
\newtheorem{example}[theorem]{Example}
\newtheorem{proposition}[theorem]{Proposition}

\newcommand{\transpose}[1]{{#1}^{\mathsf{T}}}
\newcommand{\inverse}[1]{{#1}^{-1}}
\newcommand{\mymatrix}[1]{\begin{bmatrix}#1\end{bmatrix}}
\newcommand{\Left}{\mathopen{}\mathclose\bgroup\left}
\newcommand{\Right}{\aftergroup\egroup\right}

\begin{document}

\preprint{APS/123-QED}

\title{Mixed-register Stabilizer Codes: A Coding-theoretic Perspective}

\author{Himanshu Dongre}
 \email{himanshudongre1523@gmail.com}
\affiliation{Department of Computer Science,
                    University of Illinois Chicago,
                   Chicago, Illinois, 60607}
\author{Lane G. Gunderman}
 \email{lanegunderman@gmail.com}
\affiliation{Department of Electrical and Computer Engineering,
                    University of Illinois Chicago,
                   Chicago, Illinois, 60607}

\date{\today}

\begin{abstract}
Protecting information in systems that have more than two basis states (qudits) not only offers a promising route for reducing the number of individual quantum locations that must be protected, while more accurately reflecting the structure of realistic quantum hardware, but also has some possibly enticing foundational strengths. While work in the past has largely focused on protecting information in quantum devices with locations that are some consistent local structure, this work considers coding-theoretic constraints on devices constructed from locations which may vary in their local structures---these are mixed-register quantum devices. In this work we provide some general results for mixed-register Pauli operators, then identify some stabilizer encoded information forms that are forbidden. Building on these insights, we construct coding-theoretically optimal mixed-register stabilizer codes from sets of codes defined on coprime local-dimensions. The construction of such codes results in codes with logical subspaces that do not directly correspond to any of the constituent local-dimensions.
\end{abstract}

\maketitle

\section{Introduction}

Quantum computers may enable rapid solutions to a number of computational problems. A current hurdle is the challenge of preserving information in such systems, with a particularly popular approach being stabilizer codes, which are the quantum analog of classical additive codes \cite{gottesman1996class,gottesman1997stabilizer}. These began as binary encodings, using many physical qubits to encode some logical qubits, but in time methods for handling systems with a prime, or prime-power, number of levels were developed \cite{ketkar2006nonbinary,lidar2013quantum}. These systems would require the precise control of fewer physical locations, or registers \cite{watrous2018theory}, to encode the same amount of logical information. Along the same time period methods for encoding information within bosonic registers began to be developed, including the GKP qubit encoding method \cite{gottesman2001encoding}. Since that time, a myriad of means of encoding information within bosonic modes has been explored: continuous-variable \cite{noh2020encoding,barnes2004stabilizer,braunstein1998error,lloyd1998analog}, Fock basis \cite{chuang1997bosonic,michael2016new,albert2018performance}, GKP-like encoding \cite{conrad2022gottesman,conrad2023good,gunderman2024stabilizer}, and various other forms as outlined in \cite{albert2022bosonic}. Algebraic torsion became a possible tool for discretizing codespaces that were otherwise infinite, as highlighted in the rotor code construction methods \cite{vuillot2024homological,novak2024homological}. This property is not unique to rotors, but are exhibited by mathematical rings more broadly. This then leads one down a path of composite valued registers \cite{sarkar2023qudit,gunderman2025beyond}. Recently, there has been some development of encoding information within mixed-register systems whereby not all registers have the same underlying structure nor number of levels available \cite{chakraborty2025hybrid}. This direction also has value in designing means of connecting information in a logical way that connects platforms of possibly different local structure. This work considers a similar problem, but focuses on the coding-theoretic aspects. From the results shown, we observe atypical entanglement structures and find that composite-valued registers are in many cases actually required to couple systems.

There are three different reasons for considering performing quantum computations beyond the qubit case. Firstly, the size of the computational space grows exponentially with the local-dimension of the registers as the base, thus providing a logarithmic fractional reduction in the number of registers needed for performing computations. As a particular advantage, this may reduce the cooling requirements in superconducting qubits or the number of control lasers for trapped-ion and neutral atom systems. In the most extreme case, theoretically, the register count can be reduced to a constant number of infinite registers as was shown in \cite{brenner2025factoring}. While that work demonstrated what was theoretically possible, a realistic quantum device would need to consider the practicality of protecting information in such systems. So secondly, for this reason, certain physical systems have underlying physics which are better matched with structures beyond qubits. For instance, composite-valued registers would be required to protect a device utilizing the full basis set of isotopes such as $^{87}\text{Sr}$; which has spin-$9/2$ and thus $10$ possible bases. Further, the choice for modeling the local structure can be due to the structure of the impacting noise, which is particularly focused on in bosonic error-correction, or even to simply manage leakage outside the designated qubit subspace. Lastly, there are some foundational results, such as an easy test for states being non-stabilizer \cite{gross2006hudson} and a proof of contextuality providing for the power of quantum computation \cite{howard2014contextuality}. Together, these provide multiple reasons for looking beyond qubits and better understanding the breadth of possibilities.

This paper is organized as follows. Section \ref{dfn} lays out our notation, where most matches the standards of the topic, but some have nuanced alterations which will make our later arguments simpler. Following that, section \ref{aux} shows a number of auxiliary results including a number of no-go results for mixed-register stabilizer codes. In section \ref{coding} we then turn to a construction for mixed-register codes which has tremendous flexibility and can be used to take general stabilizer codes to mixed-register codes, with only a small caveat, therefore providing good qLDPC mixed-register codes. From there we conclude and specify various workarounds that would be consistent with the results we obtained.



\section{Definitions}\label{dfn}

In this section, we lay out our definitions and special notations which will aid in our later arguments. These definitions overlap heavily with the traditional definitions for (prime-valued) qudit codes, but have some variations to enable broader analysis.

Here we refer to each distinct quantum computational location as a register \cite{watrous2018theory}, while the local-dimension specifies the number of values which the register may take. For instance, a local-dimension of $2$ indicates the register is a qubit. For the special infinite cases of continuous-variable (CV) and integer symplectic lattices (ISL), as in GKP-like encodings, we will specify which is being discussed, however, use the local-dimension $\infty$ for both cases. Throughout this work, we will typically use lowercase $q$ with a subscript to indicate that the local-dimension is some prime number, while a capital $Q$ will indicate that the local-dimension is a (possibly) composite number.

\begin{definition}
The qudit Pauli group for a register with local-dimension $Q$ is generated by the following pair of unitary operators:
\begin{equation}
    X_Q|j\rangle=|j+1\mod Q\rangle,\quad Z_Q|j\rangle=\omega_Q^j|j\rangle,\quad \omega_Q=e^{2\pi i/Q}
\end{equation}
\end{definition}
In the special cases of CV and ISL codes, the distinction is that the modulo within the $X_Q$ operator is removed, while the phase induced is replaced by any irrational power of unity. For concreteness, we select $\omega_\infty:=e^{2\pi i/\sqrt{2}}$ as our induced phase. Importantly, this has infinite order, while still acting non-trivially in the infinite. As we don't care about the global phase, so long as we are consistent, the full Pauli group on a single register, $\mathcal{P}_Q^1$, is generated by $X_Q$ and $Z_Q$. These same generators on each register generate the $n$ register Pauli group $\mathcal{P}_Q^n$. These definitions suffice as $+1$ additively generates any ring of powers, albeit in the CV case a mapping to quadrature operators is formally needed \cite{gunderman2024stabilizer}.

To augment our Pauli group, we recall the definition of the qudit Clifford group:
\begin{definition}
A unitary operator $U$ is a Clifford operator if, under conjugation, it carries a single qudit Pauli operator to a single qudit Pauli operator. More succinctly, the Clifford group is defined as:
\begin{equation}
    \{U\in SU(Q^n)|\ UPU^\dag \in \mathcal{P}_Q^n,\ \forall P\in \mathcal{P}_Q^n\}
\end{equation}
\end{definition}

In this work, we will be considering quantum computing devices and their stabilizer codes where the values of the local-dimensions of the registers do not take a single constant value but instead vary between registers. For this, we introduce the following notation to indicate the local-dimension of a subset of registers. 

A quantum computing device, $D$, is a collection of computational registers. The dimension of register $i$, $Q_i$, is given by $\mathsf{Regdim}_i(D)$. For a set of registers we use a subscript indicating the starting and ending register indices. Herein $\mathsf{Regdim}$ returns positive integer values at least $2$, indicating a qudit register, or infinity ($\infty$) indicating a bosonic register in phase space representation or an integer symplectic lattice (GKP-like) register. One could extend this notation further to also permit rotor registers or finite rotor registers---whereby the error axes take different local-dimension values.

As working directly with Pauli operators and the group of Pauli operators is often more cumbersome than working with modules of powers of the Pauli operators, we will now shift to a slightly extended version of the symplectic representation. This will most nearly mirror the $\phi_\infty$ representation from \cite{gunderman2020local,gunderman2024stabilizer}, although with some additional nuance. Importantly, we must pair our notation with that introduced by $\mathsf{Regdim}$.

\begin{definition}[$\phi$ representation of codes]
Let $P$ be a Pauli operator on an $n$ register quantum device $D$ with local-dimensions given by $\mathsf{Regdim}(D)$. Then the $\phi$ representation of this operator is given by the linear surjective map defined as carrying $\mathcal{P}_{\mathsf{Regdim}(D)}^n\mapsto \mathbb{Z}_{\mathsf{Regdim}(D)}^{2n}$. Notably entry $i$ and $i+n$ are taken modulo $\mathsf{Regdim}_i(D)$ and entry $i$ is the power of $X_Q$ and entry $i+n$ is the power of $Z_Q$ of the operator $P$. Further, this mapping is defined as a homomorphism for a pair of Pauli operators with $\phi_{\mathsf{Regdim}(D)}(s_1s_2)=\phi_{\mathsf{Regdim}(D)}(s_1)\tilde{\oplus}\phi_{\mathsf{Regdim}(D)}(s_2)$, with $\tilde{\oplus}$ being the entrywise addition with the appropriate modulos applied. Further $\phi_{\mathsf{Regdim}(D),x}$ represents the first half of the vector while $\phi_{\mathsf{Regdim}(D),z}$ represents the latter half of the vector.    
\end{definition}

For ease, we will typically write a vertical bar to separate the first and second half of the vectors as they form symplectic pairs of the same registers. In effect, the notation of a letter is a Pauli operator, whereas when wrapped with a $\phi$ the symbol is now a vector in a possibly module space. Note that formally for the infinite case we must permit the modulo to be omitted.

Now that we have a vector representation for our Pauli operators, we will define an inner product of sorts that will discern whether a pair of generalized Pauli operators commute or not.

\begin{definition}
Let $D$ be a quantum device with possibly different local-dimensions across the registers. The generalized symplectic product is defined as follows. Let $\mathsf{Regdim}(D)=[Q_1,Q_2,\ldots, Q_n]$, where $Q_i$ are the local-dimension of register $i$. Then we define the generalized symplectic inner product, $\tilde{\odot}$, as $\phi(s_i)\tilde{\odot} \phi(s_j)=\sum_{h=1}^n \frac{1}{Q_h}(\phi_{x[h]}(s_i)\phi_{z[h]}(s_j)-\phi_{x[h]}(s_j)\phi_{z[h]}(s_i))$.
\end{definition}

As a special version of this notation, we only sum the first $u$ registers if we place a subscript $u$ beneath the generalized symplectic product. This definition extends the traditional symplectic product, which indicates the power of the principle root of unitary that is induced from the reordering of operators $s_i,s_j$. The extension here is that we now are asking what the net power, over the real numbers, of the root of unitary that is induced from swapping the order of the operators. This is accomplished with the division by the local-dimension for each register, $Q_h$, thus accounting for which root of unitary is induced from the swapping of each register in the tensor product.

The central goal of this work is to consider mixed-register stabilizer codes from a coding-theoretic perspective. As such we explicitly define what suffices for such codes:

\begin{definition}
A mixed-register stabilizer code $S$ on a device $D$, with $\mathsf{Regdim}_{1\Rightarrow n}(D)=[Q_1,Q_2,\ldots, Q_n]$ where each $Q_i$ can be any positive integer (greater than $1$) or infinity, is a commuting subset of generalized Pauli operators over dimensions $[Q_1,Q_2,\ldots,Q_n]$ where commutation is satisfied iff $\phi(s_i)\tilde{\odot}\phi(s_j)=0\mod 1$ for all $s_i,s_j\in S$.
\end{definition}


A trivial way to construct mixed-register stabilizer codes is to select disjoint stabilizer codes on different local-dimensions, effectively describing the tensor product of the separate codes. For our purposes, we wish to seek out mixed-register stabilizer codes which are non-trivial. For this, we begin by first proving various properties for collections of mixed-register Pauli operators, then ruling out some possibilities for constructing mixed-register codes, before we provide an explicit and general construction in Theorem \ref{scannedconstruction}.

\section{No-go results and auxiliary claims}\label{aux}

We begin with a result on physical registers, which carries over to logical registers. We will consider the action of trans-dimensional entangling gates. An example of one between a qubit and a qutrit is $\mathsf{CNOT}_{2,3}=|0\rangle\langle0|\otimes I+|1\rangle\langle1|\otimes X_3$, with $X_3$ the qutrit $X$ gate. Let us set $U=\mathsf{CNOT}_{2,3}$, and consider conjugating $I_2Z_3$ by this operator. \textit{A priori} one might expect this to act as a Clifford operation, but we have:
\begin{eqnarray}
    U(I_2Z_3)U^\dag&=&|0\rangle\langle0|\otimes Z_3+|1\rangle\langle1|\otimes X_3Z_3X_3^\dag\\
    &=& |0\rangle\langle0|\otimes Z_3+\omega_3|1\rangle\langle1|\otimes Z_3\\
    &=& \frac{1}{2}(I_2+Z_2)\otimes Z_3+\frac{\omega_3}{2}(I_2-Z_2)\otimes Z_3\\
    &=& \frac{1}{2}((1+\omega_3)I_2+(1-\omega_3)Z_2)\otimes Z_3\\
    &=& \frac{1+\omega_3}{2}I_2Z_3+\frac{1-\omega_3}{2}Z_2Z_3.
\end{eqnarray}
Thus, we see that conjugating by this trans-dimensional entangling gate has taken a single Pauli into a sum of two Pauli operators---therefore it is not Clifford. This leads us to the following general result:

\begin{lemma}[Coprime trans-dimensional entangling gates are non-Clifford]
Suppose we have a pair of registers with local-dimensions $Q_1$ and $Q_2$, with $\mathsf{gcd}(Q_1,Q_2)=1$, then any unitary that generates entanglement between these registers must be non-Clifford.
\end{lemma}

\begin{proof}
Consider a Pauli $P=P_1\otimes I\in \mathcal{P}_{Q_1,Q_2}$ with $P_1\neq I$. Then, the order of $P_1$ in $\mathcal{P}_{Q_1}$ is $Q_1$. So the order of $P$ is also $Q_1$. Conjugation by $U$ preserves order (since it induces an automorphism on the Pauli group modulo phases). Hence, if $UPU^\dag=P_1'\otimes P_2'$, then the order of $P_1'\otimes P_2'$ in $\mathcal{P}_{Q_1,Q_2}$ is also $Q_1$. But the order of $P_1'\otimes P_2'$ can be written as $\mathsf{lcm}(r_1',r_2')$ where $r_1'$ and $r_2'$ are respectively the order of $P_1'$ in $\mathcal{P}_{Q_1}$ and $P_2'$ in $\mathcal{P}_{Q_2}$. For $i\in\{1,2\}$, $r_i'=1$ if $P_i'=I$, otherwise $r_i'=Q_i$. Hence, there are only 4 possible combination of $r$ values, but only one has a least-common-multiple of $Q_1$---$r_2'=1$. This combination dictates that $P$ must have identity on the second register which under conjugation is not impacted. Following the same argument for the other register, we obtain our result. Lastly, this result extends to any coprime dimension values using the same argument and iterating through each prime factor. Therefore the Clifford group is separable for mixed-registers which are coprime.   
\end{proof}

This result has further implications for encoding circuits. In addition, as transversality and universality are incompatible for stabilizer codes \cite{eastin2009restrictions,zeng2011transversality}, this prior fact would then suggest that for these varieties of devices, it may be wiser to teleport the entangling gates and implement all single logical register operations transversally. There is another way to sort of apply a logical trans-dimensional entangling gate, but between pseudo-primes encoded in composite local-dimension valued registers. For concreteness, let us begin with a qubit CSS code and a qutrit CSS code. These each automatically have transversal $\mathsf{SUM}$ gates (qudit $\mathsf{CNOT}$). Next, we use these to protect quhexes ($Q=6$), where we multiply the powers of the qubit code by $3$, and the powers of the qutrit code by $2$. This code will now encode logical qubits and qutrits within the quhexes, and will have a transversal $\mathsf{SUM}$ gate on quhexes. The caveat here is that this is not truly a mixed-register device, but instead a device with composite local-dimension. This construction is known as the pick-and-mix construction from \cite{gunderman2025beyond}, and the idea there is built upon later in this work.

As our next result, we consider minimal Pauli operators needed to satisfy a given commutation relationship. This result has been shown a number of ways, first in \cite{wilde2008optimal}, then with various reinterpretations in \cite{gunderman2023transforming,gunderman2024minimal,huber2025quasi}. The closest analog to our result is that obtained in \cite{sarkar2023qudit}, which can be interpreted as coinciding with our result if we multiply all registers by the least common multiple of all registers. However, here we explicitly consider the case where registers may have differing local-dimension values.


\begin{proposition}[Lower bound on the number of additional registers to resolve the commutator]\label{prop:resolution_lb}
Let $P$ be a set of compositionally independent Pauli operators for a quantum device $D$ with $\mathsf{Regdim}_{1\to u}(D)=[q_1,q_2,\ldots,q_u]$, then these cannot be extended into generators for a mixed-register stabilizer code with fewer than $\frac{1}{2}\mathsf{rank}([\tilde{\odot}_u])$ additional registers, where $[\tilde{\odot}_u]$ is the matrix of generalized symplectic inner product with $[\tilde{\odot}_u]_{ij}=\phi(p_i)\tilde{\odot}\phi(p_j)$ for $p_i,p_j\in P$.
\end{proposition}

\begin{proof}
Suppose $\lvert P\rvert=k$. Let $G=\mymatrix{G_X&G_Z}$ be the $k\times 2u$ (symplectic) check matrix for the Pauli operators in $P$ (the $i^{\text{th}}$ row of $G$ is $\phi(p_i)$ for $p_i\in P$). Let $D$ be the $u\times u$ matrix $D=\mathsf{diag}\Left(\frac{1}{q_1},\ldots,\frac{1}{q_u}\Right)$, and $J$ be the $2u\times 2u$ matrix $\mymatrix{O&D\\-D&O}$. Then, by definition,
\begin{eqnarray}
[\tilde{\odot}_u] &=& GJ\transpose{G}\\
&=& \mymatrix{G_X&G_Z}\mymatrix{O&D\\-D&O}\mymatrix{\transpose{G}_X\\\transpose{G}_Z}\\
&=& \mymatrix{G_X&G_Z}\mymatrix{D\transpose{G}_Z\\-D\transpose{G}_X}\\
&=& G_XD\transpose{G}_Z-G_ZD\transpose{G}_X.
\end{eqnarray}
To resolve this commutator, suppose we introduce $c$ registers of local-dimensions $d_1,d_2,\ldots,d_c$. Then, the blocks $G_X$ and $G_Z$ get augmented with $c$ columns each:
\begin{gather}
G_X\gets\mymatrix{G_X&\Gamma_X},\\[1mm]
G_Z\gets\mymatrix{G_Z&\Gamma_Z},
\end{gather}
where $\Gamma_X$ and $\Gamma_Z$ have dimensions $k\times c$ each. Hence, the new check matrix $G^\prime=\mymatrix{G_X&\Gamma_X&G_Z&\Gamma_Z}$. Let $D^\prime$ be the $c\times c$ matrix $D^\prime=\mathsf{diag}\Left(\frac{1}{d_1},\ldots,\frac{1}{d_c}\Right)$, and $J^\prime$ be the $2(u+c)\times 2(u+c)$ matrix $J^\prime=\mymatrix{&&D&\\&&&D^\prime\\-D&&&\\&-D^\prime&&}$. Then, by definition, the new commutator matrix is
\begin{eqnarray}
[\tilde{\odot}_u]^\prime &=& G^\prime J^\prime\transpose{G^\prime}\\
&=& \mymatrix{G_X&\Gamma_X&G_Z&\Gamma_Z}\mymatrix{&&D&\\&&&D^\prime\\-D&&&\\&-D^\prime&&}\mymatrix{\transpose{G}_X\\\transpose{\Gamma}_X\\\transpose{G}_Z\\\transpose{\Gamma}_Z}\\
&=& \Left(G_XD\transpose{G}_Z-G_ZD\transpose{G}_X\Right)+\Left(\Gamma_XD^\prime\transpose{\Gamma}_Z-\Gamma_ZD^\prime\transpose{\Gamma}_X\Right)\\
&=& [\tilde{\odot}_u]+\Left(\Gamma_XD^\prime\transpose{\Gamma}_Z-\Gamma_ZD^\prime\transpose{\Gamma}_X\Right).
\end{eqnarray}
Since our aim in the first place was to resolve the commutator, set $[\tilde{\odot}_u]^\prime=O$. Hence,
\begin{eqnarray}
&&\Gamma_ZD^\prime\transpose{\Gamma}_X-\Gamma_XD^\prime\transpose{\Gamma}_Z = [\tilde{\odot}_u]\\
&\Rightarrow& \mathsf{rank}\Left(\Gamma_ZD^\prime\transpose{\Gamma}_X-\Gamma_XD^\prime\transpose{\Gamma}_Z\Right) = \mathsf{rank}([\tilde{\odot}_u]).\label{eq:rank_eq}
\end{eqnarray}
Since multiplication by an invertible diagonal matrix (like $D^\prime$) simply scales the entries; i.e., preserves the column and row subspaces,
\begin{gather}
\mathsf{rank}\Left(\Gamma_ZD^\prime\transpose{\Gamma}_X\Right)=\mathsf{rank}\Left(\Gamma_Z\transpose{\Gamma}_X\Right).
\end{gather}
By Sylvester's rank inequality,
\begin{gather}
\mathsf{rank}(\Gamma_Z\transpose{\Gamma}_X)\leq \min(\mathsf{rank}(\Gamma_X),\mathsf{rank}(\Gamma_Z)).
\end{gather}
Since both $\Gamma_X$ and $\Gamma_Z$ are $k\times c$ matrices, both $\mathsf{rank}(\Gamma_X)$ and $\mathsf{rank}(\Gamma_Z)$ are constrained by the dimension $c$. Hence,
\begin{gather}
\min(\mathsf{rank}(\Gamma_X),\mathsf{rank}(\Gamma_Z))\leq c.
\end{gather}
Therefore,
\begin{eqnarray}
\mathsf{rank}(\Gamma_ZD^\prime\transpose{\Gamma}_X) &=& \mathsf{rank}(\Gamma_Z\transpose{\Gamma}_X)\\
&\leq& \min(\mathsf{rank}(\Gamma_X),\mathsf{rank}(\Gamma_Z))\\
&\leq& c.\label{eq:half_rank_bound}
\end{eqnarray}
Since $\Gamma_XD^\prime\transpose{\Gamma}_Z=\transpose{(\Gamma_ZD^\prime\transpose{\Gamma}_X)}$, $\mathsf{rank}\Left(-\Gamma_XD^\prime\transpose{\Gamma}_Z\Right)=\mathsf{rank}\Left(\Gamma_XD^\prime\transpose{\Gamma}_Z\Right)=\mathsf{rank}\Left(\Gamma_ZD^\prime\transpose{\Gamma}_X\Right)$. By rank subadditivity and \eqref{eq:half_rank_bound},
\begin{gather}
\mathsf{rank}\Left(\Gamma_ZD^\prime\transpose{\Gamma}_X-\Gamma_XD^\prime\transpose{\Gamma}_Z\Right)\leq 2c.\label{eq:rank_bound}
\end{gather}
Therefore, from \eqref{eq:rank_eq} and \eqref{eq:rank_bound},
\begin{gather}
c\geq\frac{1}{2}\mathsf{rank}([\tilde{\odot}_u]).
\end{gather}
\end{proof}

Next, we provide a canonical structural decomposition for collections of Pauli operators on mixed-register quantum devices. This will enable our first mixed-register code construction as given in Construction \ref{constr:resolution}.

\begin{theorem}[Mixed-register Pauli subgroup decomposition]\label{thm:decomposition}
Suppose $\mathcal{D}$ is a quantum device with $n$ registers having local-dimensions $q_1,\ldots,q_n \geq 2$, and $\hat{\mathcal{P}}$ is the mixed-register Pauli group mod phases on $\mathcal{D}$. Let $Q=\mathsf{lcm}(q_1,\ldots,q_n)$. Let $\mathcal{T}\leq\hat{\mathcal{P}}$ be any subgroup generated by finitely many Paulis. Then there exists a generating set for $\mathcal{T}$ of the form $$\{W_1,W_2,\ldots,W_\ell,U,V,U_2,V_2,\ldots,U_c,V_c\}$$ such that the following conditions are satisfied:
\begin{enumerate}
    \item Commutation condition: the generators satisfy the following relations (all mod 1):
    \begin{enumerate}
        \item $\phi(W_i)\tilde{\odot}\phi(P)=0$ for all $P\in\mathcal{T}$ for all $i\in\{1,\dots,\ell\}$,
        \item $\phi(U_i)\tilde{\odot}\phi(U_j)=\phi(V_i)\tilde{\odot}\phi(V_j)=0$ for all $i,j\in\{1,\dots,c\}$,
        \item $\phi(U_i)\tilde{\odot}\phi(V_j)=0$ for all $i\ne j$,
        \item $\phi(U_i)\tilde{\odot}\phi(V_i)=\frac{1}{d_i} (\!\!\!\!\mod 1)$ with $d_i\mid Q$.
    \end{enumerate}
    \item Divisibility chain condition: $d_1\mid d_2\mid\ldots\mid d_c$.
    \item Uniqueness condition: the multiset $\{d_1,\ldots.d_c\}$ is unique depending only on $\mathcal{T}$.
\end{enumerate}
\end{theorem}

Due to the length of the proof, we have placed this result within the appendix. This result provides a standard way of decomposing Pauli operators even with mixed-register devices. It essentially "factors out" the symplectic non-commutativity from the subgroup and expresses it in the form of an external direct product of cyclic groups which occur in pairs (a pair being characterized by equal order). A pair of generators that generates the cyclic groups in a pair is called a hyperbolic pair. The final decomposition result in \eqref{eq:main_decomp} is that the subgroup decomposes into a completely commuting "isotropic" part, and a "symplectic" part generated by hyperbolic pairs, written as an external direct product of cyclic groups.

\subsection{No-go results}

Here we show a pair of no-go results on mixed-register stabilizer codes. These rule out the possibility of non-trivially encoding oscillators, or rotors, with finite registers, as well as the possibility of encoding coprime local-dimensional codes. We begin by showing the following result that has likely been shown before:
\begin{lemma}\label{supergens}
Let $g$ be the generators for the stabilizer code $S$, and let $s$ be a superset of $g$ that can also generate $S$. Then the following projectors are all equal: $\Pi_S=\Pi_g=\Pi_s$.
\end{lemma}

\begin{proof}
Note that $\Pi_S=\Pi_g$ is by definition true from $\Pi_g=\prod_{i=1}^{|g|}\frac{1}{\mathsf{ord}(g_i)}\sum_{j=0}^{\mathsf{ord}(g_i)-1}g_i^j$ being a projector onto the $+1$ eigenspace, which is $\Pi_S$. The only new portion is verifying equality with $\Pi_s$. Next, let there be redundant generators so that $g\subset s\subset S$, then $\Pi_s=\prod_{i=1}^{|s|}\frac{1}{\mathsf{ord}(s_i)}\sum_{j=0}^{\mathsf{ord}(s_i)-1}s_i^j=\Pi_g$ as each term in the product, $\frac{1}{\mathsf{ord}(s_i)}\sum_{j=0}^{\mathsf{ord}(s_i)-1}s_i^j$, is itself a projector, so we may break each redundant element into the product of generators, and due to the commuting nature and projectors squaring to themselves, reduces back to $\Pi_g$.
\end{proof}

As a trivial example, if a code has disjoint subsets then this is equivalent to two separate stabilizer codes. Of particular note then is the case where these subsets have some nontrivial overlap between them, indicating entanglement betwixt the parts of the system. As a first use of this tool we will show that true oscillator-qudit mixed-register codes must be disjoint subsets, thus meaning that a stabilizer code cannot encode logical oscillators (or rotors) along with qudits. Instead the oscillators must become finitely quantized, or a broader framework such as subsystem codes or approximate codes must be employed.

For a true oscillator or rotor, the order of the operators must be unbounded, so that we emulate the structure of $\mathbb{R}$ or $\mathbb{Z}$. Further, we must have the operators in the limit of infinity still induce a nontrivial phase. To satisfy both of these conditions we may treat the local-dimension as being some irrational value multiplied by an integer. This is isomorphic to the integers and dense in the real numbers. For concreteness we will pick a local-dimension of $\frac{1}{\sqrt{2}}$.

\begin{lemma}[No true oscillator-qudit codes]
Let $S$ be a mixed-register stabilizer code with registers $1$ to $m$ having local-dimension corresponding to either $\mathbb{R}$ or $\mathbb{Z}$, while registers $m+1$ to $n$ have finite local-dimension values. Then $S$ must be constructed from a disjoint code on registers $1$ to $m$, which must be separable from the remaining registers.
\end{lemma}

\begin{proof}

By supposition there must be at least one pair of generators which can detect a single infinite $X$ and $Z$ operator. For the sake of contradiction, let us assume that the code is not disjoint and thus must have a trivial total generalized symplectic product. Given this, using the generalized symplectic product, $\tilde{\odot}$, we compute the symplectic product between the generators, which must include some integer multiple of $\frac{1}{\sqrt{2}}$. However, all finite valued local-dimensions will only contribute rational numbers. From this, we cannot generate generate a symplectic product that is $0$ as needed in the Pauli representation, nor an integer multiple if using the Heisenberg operator representation. This contradicts this assumption, therefore the code must be disjoint.
\end{proof}

Summary: $\Pi_{\text{mixed}}=\Pi_\infty\otimes \Pi_{\text{finite}}$, thus when describing the logical space, we can describe it as a pair of values $k_\infty$, the number of infinite sized registers encoded, and $K_{\text{finite}}$ which is the total number of orthonormal logical bases. An important caveat to the prior that this does not forbid using discretized versions of these infinite spaces, such as using a finite-GKP-like set of registers along with traditional qudit registers. A further impact of this is that if rotor registers are used for the code, the local-dimension of these registers would need to be a factor of the least common multiple of the other registers, thus putting a limit on the phase space separation possible between codewords.

\begin{lemma}[Coprime codes must be disjoint]
Let $S$ be a mixed-register code with local-dimensions on the first $m_1$ registers having local-dimension $Q_1$, the next $m_2$ registers having local-dimension $Q_2$, and so forth, with each local-dimension being coprime to all others. Then $S$ is a concatenation of a length $m_1$ code with local-dimension $Q_1$ with a length $m_2$ code with local-dimension $Q_2$ and so forth. Therefore forbidding a true mixed coprime-dimensional stabilizer code.
\end{lemma}

Use caution with the prior, as we are using concatenated to mean one after the other (string concatenation, so separable tensor products). When we intend a code's logical registers as input to another code, the term composition of codes will be used. Further, note that this prior result is a little more general than as would first seem. While a proof with all dimensions being primes would be sufficient, this goes beyond and rules out composites if they are all coprime.

\begin{proof}
Let the symplectic representation of this code be given by $\phi_\infty(S)$, equipped with the generalized symplectic product operator $\tilde{\odot}$. Further, define $\mathcal{Q}=\mathsf{lcm}(Q_1,Q_2,\ldots Q_t)$ as the least common multiple of all the local-dimensions being used. We begin with the case of two coprime local-dimensions $Q_1$ and $Q_2$. We augment the matrix $\phi_\infty(S)$ with redundant rows given by $Q_1\phi_\infty(S)$, which will by definition eliminate the support within the first $m_1$ registers. We then augment the code further with $Q_2\phi_\infty(S)$, which will by definition eliminate support within the $m_2$ registers. From lemma \ref{supergens}, these checks still generate the same codespace, but, further, notice that these form a generating set for the original code $S$. We may then eliminate from our augmented code the original generators and use $Q_1\phi_\infty(S)$ and $Q_2\phi_\infty(S)$ as generators for the codespace. Importantly, these have disjoint support, and thus the projector onto the codespace will break into a tensor product over these registers---implying the code is a $m_1$ length $Q_1$ local-dimension stabilizer code followed by a $m_2$ length $Q_2$ local-dimension stabilizer code. In short the code has been transformed into the following form:
\begin{equation}
    \phi(S)=\begin{bmatrix}
        S_{Q_1,x}& 0&|&S_{Q_1,z} & 0\\
        0 & S_{Q_2,x}&|&0&S_{Q_2,z}
    \end{bmatrix}\equiv \phi(S_{Q_1})\otimes \phi(S_{Q_2}).
\end{equation}
To handle more than two coprimes, instead of multiplying by $Q_1,Q_2$, we multiply by each of $\mathcal{Q}/Q_i$, from which the result carries the same reasoning.
\end{proof}

\section{Coding-theoretic results}\label{coding}

Before we dive into our new results related to constructing mixed-register stabilizer codes, there are a couple of nuances we wish to emphasize. To start, recognize that in the process of trivializing the generalized symplectic product, a common denominator can be formed (the least common multiple), which may suggest that the local-dimension of some registers must be the least common multiple. This is not quite so. Consider the following example $\langle X_6^3Z_5^3,Z_6X_5\rangle$. These don't commute so the further register(s) will need to fix that. The generalized symplectic product has absolute value $1/10$. While we could use a $30$ level register, that would have a torsion that would generate undetectable errors (multiples of $10$), so we must instead use a $10$ level register. Generally, if the numerator happens to be a divisor of the least common multiple, we generate nontrivial torsion. So what we actually have is that the registers must have $Q$ such that $Q$ divides the LCM. For the same reason, selecting a multiple of the LCM for the local-dimension would be permissible, but again would generate distance reducing torsion. Let us begin with a small example of a code to ground our further discussions.

\begin{example}
    As a simple example we consider the code generated by $\langle X_2X_6^3I,IX_6^2X_3\rangle$, which can \textit{detect} a single phase error. This has logical operators given by $\bar{X}=\{X_2II,IIX_3\}$, or simply $\bar{X}=IX_6I$, which suffices from Bezout's identity for coprime numbers. The logical phase operators are given by $\bar{Z}=\{Z_2Z_6^{-3}I,IZ_6^{-2}Z_3\}$. A member of the logical space is given by \begin{equation}
|'0'\rangle_L=\frac{1}{\sqrt{6}}(|000\rangle+|130\rangle+|021\rangle+|042\rangle+|151\rangle+|112\rangle),    
\end{equation}
With the other codewords generated by the logical $X$ operators---here $IX_6I$ powers suffice. Notably this vector is not separable within any of the constituent register values, thus they are entangled. Notably we labeled the logical state as $'0'$ as it does not fully correspond to a qubit, qutrit, nor quhex state. 
\end{example}

Given that we are considering codes with registers of possibly differing dimensions and Pauli operators of possibly differing orders, we need a more general result for computing the total computational space available.

\begin{lemma}\label{mixedlog}
    Let $D$ be a mixed-register quantum device. Let $S$ be the generators for a stabilizer code on $D$ with only finite local-dimensions for each value of $\mathsf{Regdim}(D)$. Then the full logical subspace has $K$ orthonormal bases with $K=\prod_{i=1}^n \mathsf{Regdim}_i(D)/\mathcal{D}$ where $\mathcal{D}=\prod_{p_i\in \mathsf{Regdim}(D)}p_i^{|\{s\in S:\ s^{p_i}=I\}|}$.
\end{lemma}

This is the mixed-register analog of the traditional finite dimensional computation of the logical subspace. For a uniform, even if composite, code the expression is more comprehensible \cite{gunderman2025beyond}.

\begin{proof}
Let $\mathcal{S}$ be the full stabilizer code. Then the projector onto the $+1$ eigenspace of the stabilizer group, and so the codespace, is given by:
\begin{equation}
    \Pi_\mathcal{S}=\frac{1}{\mathcal{D}}\sum_{s\in \mathcal{S}} s ,
\end{equation}
with $\mathcal{D}=\prod_{p_i\in \mathsf{Regdim}(D)}p_i^{|\{s\in S:\ s^{p_i}=I\}|}$ accounting for the normalization factor from the differing possible orders for the generators. This satisfies $\Pi_{\mathcal{S}}^2=\Pi_{\mathcal{S}}$. The dimension is then computed from the trace of this operator:
\begin{eqnarray}
\mathsf{tr}(\Pi_\mathcal{S})&=& \frac{1}{\mathcal{D}}\mathsf{tr}\Left(\sum_{s\in \mathcal{S}} s\Right)\\
&=& \frac{1}{\mathcal{D}} \mathsf{tr}(I)\\
&=& \frac{1}{\mathcal{D}} \prod_{i=1}^n \mathsf{Regdim}_i(D)=K,
\end{eqnarray}
where the second equality follows from all of the Pauli operators being traceless, except for the identity operator.
\end{proof}

We are now ready for our first construction, which is based on our canonical structural decomposition Theorem \ref{thm:decomposition}. Perhaps the simplest way to construct a mixed-register stabilizer code is to extend a collection of possibly non-commuting Pauli operators onto more registers and resolving all of the commutators. We take this approach for our first construction.

\bigskip
\refstepcounter{theorem}
\noindent\textbf{Construction~\thetheorem}~(Non-commutativity resolution).\label{constr:resolution}
A quintessential application of Theorem~\ref{thm:decomposition} would be to construct a mixed-register stabilizer code from a mixed-register Pauli subgroup with a nontrivial symplectic product matrix. Suppose we have a system like in the setup of Theorem~\ref{thm:decomposition}: A quantum device $\mathcal{D}$ with $n$ registers having local-dimensions $q_1,\ldots,q_n \geq 2$, the mixed-register Pauli group mod phases on $\mathcal{D}$, $\hat{\mathcal{P}}_n$, $Q=\mathsf{lcm}(q_1,\ldots,q_n)$, and some subgroup $\mathcal{T}\leq\hat{\mathcal{P}}_n$ generated by finitely many Paulis. Then, by Theorem~\ref{thm:decomposition}, we can write a generating set for $\mathcal{T}$ of the form
\begin{gather}
\{W_1,W_2,\dots,W_\ell,U_1,V_1,U_2,V_2,\dots,U_c,V_c\}.\label{eq:generators_list}
\end{gather}
If $c=0$, then $\mathcal{T}$ is already a mixed-register stabilizer code. Assume $c\ne 0$. Then the idea is to "resolve" the non-commutativity by introducing additional registers as follows:
\begin{enumerate}
    \item Suppose $\phi(U_i)\tilde{\odot}\phi(V_i)=\frac{1}{d_i}\mod 1,\text{ where }d_i\mid Q$. Then, augment $\mathcal{D}$ with $c$ registers having local-dimensions $d_1,d_2,\dots,d_c$.
    \item Suppose $C$ refers to the subsystem constituted by the $c$ newly introduced registers. Modify the generators listed in \eqref{eq:generators_list} to account for $C$ as follows: Augment each $W_i$ with the identity on $C$, augment each $U_i$ with the $X_{d_i}$ Pauli on the $i^{\text{th}}$ register in $C$ and the identity on the rest, augment each $V_i$ with the $\inverse{Z}_{d_i}$ Pauli on the $i^{\text{th}}$ register in $C$ and the identity on the rest. Define the new generators
    \begin{eqnarray}
        \bar{W}_i&=&W_i\otimes I,\\
        \bar{U}_i&=&U_i\otimes X_{d_i},\\
        \bar{V}_i&=&V_i\otimes\inverse{Z}_{d_i}.
    \end{eqnarray}
\end{enumerate}
The new generators satisfy the following relations (all mod 1):
\begin{enumerate}
    \item $\phi\Left(\bar{W}_i\Right)\tilde{\odot}\phi\Left(\bar{W}_j\Right)=0$ for all $i,j\in\{1,\dots,\ell\}$,
    \item $\phi\Left(\bar{W}_i\Right)\tilde{\odot}\phi\Left(\bar{U}_j\Right)=\phi\Left(\bar{W}_i\Right)\tilde{\odot}\phi\Left(\bar{V}_j\Right)=0$ for all $i\in\{1,\dots,\ell\}$, $j\in\{1,\dots,c\}$,
    \item $\phi\Left(\bar{U}_i\Right)\tilde{\odot}\phi\Left(\bar{U}_j\Right)=\phi\Left(\bar{V}_i\Right)\tilde{\odot}\phi\Left(\bar{V}_j\Right)0$ for all $i,j\in\{1,\dots,c\}$,
    \item $\phi\Left(\bar{U}_i\Right)\tilde{\odot}\phi\Left(\bar{V}_j\Right)=0$ for all $i\ne j$.
\end{enumerate}
These commutations are trivial because they follow directly from the commutation relations of the generators of $\mathcal{T}$ and the fact that the identity commutes with everything. The key commutation relation which shows the resolved non-commutativity is
\begin{eqnarray}
\phi\Left(\bar{U}_i\Right)\tilde{\odot}\phi\Left(\bar{V}_i\Right)&=&\phi(U_i)\tilde{\odot}\phi(V_i)+\frac{1}{d_i}(-1-0)\\
&=&\frac{1}{d_i}-\frac{1}{d_i}\\
&=&0.
\end{eqnarray}
Therefore, consider the mixed-register Pauli subgroup on $(n+c)$ registers generated by the new generators
\begin{gather}
\mathcal{S}=\Left\langle\bar{W}_1,\bar{W}_2,\dots,\bar{W}_\ell,\bar{U}_1,\bar{V}_1,\bar{U}_2,\bar{V}_2,\dots,\bar{U}_c,\bar{V}_c\Right\rangle\leq\hat{\mathcal{P}}_{n+c}.
\end{gather}
Then, $\mathcal{S}$ is a mixed-register stabilizer code.

\begin{remark}[Optimality of the above construction in terms of efficiency]\label{rmk:decomp_efficiency}
It is important to comment on the "efficiency" of this construction; i.e., whether a scenario is possible where one follows this construction and ends up adding two registers---a qubit and a qutrit---when they could've achieved the same outcome with one quhex and figuratively killed two birds with one stone. However, fortunately, such a scenario is not possible. The theorem guarantees that we add the minimum number of registers required to resolve the commutation, which is optimal. This is because of the divisibility chain condition (point (2)) in the statement of Theorem~\ref{thm:decomposition}. Suppose it does end up happening that the construction yields the addition of one qubit and one qutrit register; i.e., $d_i=2$ for some $i\in\{1,\dots,c-1\}$ and $d_j=3$ for some $j>i$ without loss of generality. Then, $d_i\nmid d_j$. Hence, this violates the divisibility chain condition. The qubit and the qutrit can instead be "amalgamated" into one quhex register to actually give a divisibility chain with $d_i=1$, $d_j=6$, and $d_i\mid d_j$. Consequently, the $d_i=1$ is omitted and we are left with one quhex instead of two (qubit and qutrit) registers. Here, the word "amalgamated" abstracts the precise mathematical reason this is guaranteed, which needs context from the proof, that the external direct product of a cyclic group of order 2 and a cyclic group of order 3 is isomorphic to a cyclic group of order 6. Hence, the invariant-factor form of the cyclic-group decomposition in \eqref{eq:final_K_cylic_decomp} replaces $\mathbb{Z}_2\oplus\mathbb{Z}_3$ with $\mathbb{Z}_6$. This, of course, holds for any two coprime dimensions. After all such amalgamations are made, we are left with a unique (by the uniqueness condition in Theorem~\ref{thm:decomposition}) multiset $\{d_1,d_2,\dots,d_c\}$ that is the most efficient for our $\mathcal{T}$.   
\end{remark}

Remark~\ref{rmk:decomp_efficiency} claims that Construction~\ref{constr:resolution} guarantees we add the minimum number of registers required to resolve the non-commutativity. Moreover, Proposition~\ref{prop:resolution_lb} provides an explicit general lower bound on the number of registers one must add to resolve the commutator. A natural question then becomes by how much does the minimum number of additional registers guaranteed by Construction~\ref{constr:resolution} exceed the hard lower bound given by Proposition~\ref{prop:resolution_lb}. The next result asserts that Construction~\ref{constr:resolution} in fact saturates the lower bound in Proposition~\ref{prop:resolution_lb}.

\begin{proposition}
Suppose $\mathcal{D}$ is a quantum device with $n$ registers having local-dimensions $q_1,\ldots,q_n \geq 2$, and $\hat{\mathcal{P}}$ is the mixed-register Pauli group mod phases on $\mathcal{D}$. Let $\mathcal{T}\leq\hat{\mathcal{P}}$ be any subgroup generated by finitely many Paulis. Let $\mathcal{S}$ be the mixed-register stabilizer code constructed from $\mathcal{T}$ using Construction~\ref{constr:resolution}, and suppose $c$ is the number of additional registers. Then,
\begin{gather}
c=\frac{1}{2}\mathsf{rank}[\tilde{\odot}],
\end{gather}
where $[\tilde{\odot}]$ is the generalized symplectic inner product matrix for $\mathcal{T}$ using the generating set from the construction.
\end{proposition}

\begin{proof}
A central part of Construction~\ref{constr:resolution} is writing a generating set for $\mathcal{T}$ of the form
\begin{gather}
\{W_1,W_2,\dots,W_\ell,U_1,V_1,U_2,V_2,\dots,U_c,V_c\}
\end{gather}
using Theorem~\ref{thm:decomposition}. Using this generating set, it is most intuitive to view $[\tilde{\odot}]$ as a block matrix:
\begin{gather}
[\tilde{\odot}]=
    \begin{blockarray}{cccc}
    & W_j\mathrm{s} & U_j\mathrm{s} & V_j\mathrm{s} \\
        \begin{block}{c[ccc]}
            W_i\mathrm{s} & O_{\ell\times\ell} & O_{\ell\times c} & O_{\ell\times c} \\
            U_i\mathrm{s} & O_{c\times\ell}    & O_{c\times c}    & D_{c\times c} \\
            V_i\mathrm{s} & O_{c\times\ell}    & -D_{c\times c}   & O_{c\times c} \\
        \end{block}
    \end{blockarray}\label{eq:block_view}
\end{gather}
This form for $[\tilde{\odot}]$ comes about because of the commutation relations of the generators as specified in Theorem~\ref{thm:decomposition}. Everything except the $D$ and $-D$ blocks are null, where $D=\mymatrix{\frac{1}{d_1}&&\\&\ddots&\\&&\frac{1}{d_c}}$ is a $c\times c$ \emph{diagonal} matrix because the generators in the hyperbolic pairs have a nontrivial symplectic product only with each other, and 0 with everything else. Hence, it is clear from \eqref{eq:block_view} and the fact that $D$ is diagonal that exactly $2c$ out of the $\ell+2c$ columns in $[\tilde{\odot}]$ are not all zero. Each of the $2c$ columns, in fact, has exactly one nonzero, and each nonzero is in its separate row. Hence, these $2c$ columns are the linearly independent (pivot) columns of the matrix. Therefore,
\begin{eqnarray}
&&\mathsf{rank}([\tilde{\odot}])=2c\\
&\Rightarrow& c=\frac{1}{2}\mathsf{rank}([\tilde{\odot}]).
\end{eqnarray}
\end{proof}

The prior construction generated shorter mixed-register stabilizer codes, however, it does not fully provide for guarantees on the distance of the constructed codes. The following construction does precisely that, however, it requires seeding with a pair of coprime dimensional stabilizer codes and derives its distance directly from the seeding codes.

\begin{theorem}[Mixed-register scanned construction]\label{scannedconstruction}
Let $S_1$ be a code with parameters $[[n_1,k_1,d_1]]_{Q_1}$ and let $S_2$ be a code with parameters $[[n_2,k_2,d_2]]_{Q_2}$ for a choice of coprime numbers $Q_1,Q_2$. Further, let $\mathcal{I}$ and $\mathcal{J}$ be any (non-disjoint) subsets of $n=n_1+n_2$. Then we may construct a mixed-register code with local-dimension $Q_1$ on registers in $\mathcal{I}-(\mathcal{I}\cap\mathcal{J})$, local-dimension $Q_2$ on registers in $\mathcal{J}-(\mathcal{I}\cap\mathcal{J})$, and local-dimension $Q:=Q_1Q_2$ on registers in $\mathcal{I}\cap\mathcal{J}$. Further, this code will have parameters:
\begin{equation}
[[n=n_1+n_2-|\mathcal{I}\cap\mathcal{J}|,K=Q_1^{k_1}Q_2^{k_2},d=\min(d_1,d_2)]]_{\{Q_1,Q_1Q_2,Q_2\} }
\end{equation}
\end{theorem}


\begin{proof}
First we verify that the generalized symplectic products ensure that the generators for the code all induce zero net phase, and thus commute. We leave the powers alone on $\mathcal{I}-\mathcal{I}\cap\mathcal{J}$ and $\mathcal{J}-\mathcal{I}\cap\mathcal{J}$, but those on $\mathcal{I}\cap\mathcal{J}$ we multiply the generators from $S_1$ by $Q/Q_2$ so that they induce a phase of is the $Q_1$ root of unity (generally we multiply by, with $\mathcal{Q}=\mathsf{lcm}(Q_1,Q_2,\ldots)$, $\mathcal{Q}/Q_i$ so that the additive order over $\mathcal{Q}$ is $Q_i$ still). Then these will still retain commutation within $Q_1$ generators, likewise for the generators originating from the $Q_2$ code. The key property is that these powers will cancel out over $Q$ with those brought from $S_2$ in the overlap region as the product of the powers is congruent to $0$. This will then generate a commuting set of generators for a mixed-register stabilizer code.

For the size of the logical space, 1) here it really does not make sense to take the log of $K$ to write a $k$ expression since the logical space is not of dimension $Q_1,Q_2,Q_1Q_2$, 2) this follows from the normal subspace argument, given in lemma \ref{mixedlog} and noting that the order of the generators will not be altered.

For the distance, this requires a little more care. Clearly on $\mathcal{I}-\mathcal{I}\cap\mathcal{J}$ those are protected by the starting $S_1$. The slightly tricky thing is noticing that for the $Q$ registers, 1) if the error has an additive term that is a multiple of $\mathcal{Q}/Q_2$, then that is handled by the original $S_1$ code, while if it has an additive term that is a multiple of $\mathcal{Q}/Q_1$ then it is handled by $S_2$, so the distance is preserved. This is why it is essential that $Q_1,Q_2$ are coprime so that by Bezout's identity we can decompose any error into a unique sum of these powers.
\end{proof}

Notice the nuance here that these are not pure $Q_1$ nor $Q_2$ nor $Q_1Q_2$ registers! This was precisely what we saw in the opening example. A slightly surprising thing here is that the partitioning can be \textit{arbitrary}. If they have a trivial overlap, then the codes are just concatenated. But with this construction we can decide how much we want the codes to overlap, but a single register is sufficient so that the codespace projector is not simply the tensor product of disjoint spaces. One could inquire what would happen if the overlapping registers are a multiple of the least common multiple. In that case we could only correct errors that are multiples of the greatest common divisor of this value and the true least common multiple---thus ruining the distance of the code. As a concrete connection with some prior work, this construction perfectly matches the pick-and-mix construction from \cite{gunderman2025beyond} in the case that all registers are the same composite number and that the codes are repeated so that the total length is uniform. As another possible interpretation, we could consider this whole construction as making stabilizer codes over the least common multiple local-dimension, but demanding that errors only be powers which are multiples of the least common multiple divided by the local-dimension---however, that is highly unphysical. 

\begin{corollary}
The prior construction can be extended to the intersection of many coprime stabilizer codes and can involve many codes intersecting at a subset of registers.
\end{corollary}

Another important aspect to the prior theorem is that this provides for good qLDPC mixed-register codes. If one wished to extend this to arbitrary choices of local-dimensions, while still coprime, the local-dimension-invariant framework would permit that \cite{gunderman2020local,gunderman2024stabilizer,gunderman2025beyond}. Another notable remark is that this property of requiring at least a single composite register that is the least common multiple is a purely quantum phenomena, in that classical codeword bases can be specified freely within their respective lattices, but due to the commutation constraint we require at least one least common multiple register.

\begin{remark}[Topological phase connection]
As a simple example, begin by recalling that the codewords of the toric code correspond to different topologically protected phases. Let us first construct a mixed-register code consisting of a qubit and a qutrit surface code which are disjoint \cite{gunderman2024stabilizer}. Then upon replacing a single edge vertex with a $Q=6$ level register, the two surfaces are joined and entangled. In particular the total topologic phases no longer correspond to a qubit and a qutrit phase, but instead now a $6$-tuple of degenerate quhex phases, with only a qubit like phase appearing on the pure qubit registers and only a qutrit like phase appearing on the pure qutrit registers, while the joint quhex registers become a value that represents their join topologic phases. Notably this is achieved by a single quhex, while if the single point-like connection between the surfaces is expanded into a full interface then a clear boundary between the phases is formed. While in this example we selected a qubit and qutrit toric code, we can select any coprime dimensions for the toric codes then join them. Or more broadly select any topologic code and perform the same merging procedure as outlined in Theorem \ref{scannedconstruction}, where the overlap between the topologic codes can be altered simply by changing which registers take values in the $\mathsf{lcm}$ of the intersecting surfaces.
\end{remark}

\section{Conclusion}

In this work we have considered mixed-register stabilizer codes, taking a more coding-theoretic perspective. We began by showing general results on mixed-register Pauli operators, followed by some no-go results on mixed-register stabilizer codes. This then led us to a first construction for mixed-register stabilizer codes which made use of a canonical decomposition for mixed-register Pauli operators. Then we provided a second construction allowing for the use of arbitrary stabilizer codes to generate such codes. This generates entanglement between varying local-dimensions, resulting in atypical topological phase structure. Important future directions directly building off of the results herein include generating other constructions for mixed-register codes, as well as further studies of possible transversal logical operations and logical operations more generally.

We also wish to point out that importantly the bulk of our results pertain to mixed-register stabilizer codes. In particular our no-go results may hold in scenarios beyond stabilizer codes. Some instances of possible workarounds include: subsystems \cite{li2025qutrit,li2026transversal}, finite embeddings of infinite registers \cite{chakraborty2025hybrid}, approximate codes, beyond Pauli stabilized codes \cite{ball2026error}, floquet codes, and non-additive codes \cite{ball2026error}, among other options. Further research into possible constraints for these systems could be fruitful, as we centered stabilizer codes here. 

\section*{Acknowledgments}

We thank Md. Shahinul Islam for helpful early discussions and pointing out the work \cite{ball2026error}.

\bibliographystyle{unsrt}
\phantomsection  
\renewcommand*{\bibname}{References}

\bibliography{main}

@article{ketkar2006nonbinary,
  title={Nonbinary stabilizer codes over finite fields},
  author={Ketkar, Avanti and Klappenecker, Andreas and Kumar, Santosh and Sarvepalli, Pradeep Kiran},
  journal={IEEE transactions on information theory},
  volume={52},
  number={11},
  pages={4892--4914},
  year={2006},
DOI={10.1109/TIT.2006.883612},
  publisher={IEEE}
}

@article{gunderman2020local,
  title={Local-dimension-invariant qudit stabilizer codes},
  author={Gunderman, Lane G},
  journal={Physical Review A},
  volume={101},
  number={5},
  pages={052343},
  year={2020},
  DOI = {10.1103/PhysRevA.101.052343},
  publisher={APS}
}

@book{watrous2018theory,
  title={The theory of quantum information},
  author={Watrous, John},
  year={2018},
DOI={10.1017/9781316848142},
  publisher={Cambridge university press}
}

@article{gottesman2001encoding,
  title={Encoding a qubit in an oscillator},
  author={Gottesman, Daniel and Kitaev, Alexei and Preskill, John},
  journal={Physical Review A},
  volume={64},
  number={1},
  pages={012310},
  year={2001},
DOI={10.1103/PhysRevA.64.012310},
  publisher={APS}
}

@article{lloyd1998analog,
  title={Analog quantum error correction},
  author={Lloyd, Seth and Slotine, Jean-Jacques E},
  journal={Physical Review Letters},
  volume={80},
  number={18},
  pages={4088},
  year={1998},
DOI={10.1103/PhysRevLett.80.4088},
  publisher={APS}
}

@incollection{braunstein1998error,
  title={Error correction for continuous quantum variables},
  author={Braunstein, Samuel L},
  booktitle={Quantum Information with Continuous Variables},
  pages={19--29},
  year={1998},
DOI={10.1103/PhysRevLett.80.4084},
  publisher={Springer}
}

@article{noh2020encoding,
  title={Encoding an oscillator into many oscillators},
  author={Noh, Kyungjoo and Girvin, SM and Jiang, Liang},
  journal={Physical Review Letters},
  volume={125},
  number={8},
  pages={080503},
  year={2020},
DOI={10.1103/PhysRevLett.125.080503},
  publisher={APS}
}

@article{barnes2004stabilizer,
  title={Stabilizer codes for continuous-variable quantum error correction},
  author={Barnes, Richard L},
  journal={arXiv preprint quant-ph/0405064},
DOI={10.48550/arXiv.quant-ph/0405064},
  year={2004}
}

@article{albert2022bosonic,
  title={Bosonic coding: introduction and use cases},
  author={Albert, Victor V},
  journal={arXiv preprint arXiv:2211.05714},
DOI={10.48550/arXiv.2211.05714},
  year={2022}
}

@book{gottesman1997stabilizer,
  title={Stabilizer codes and quantum error correction},
  author={Gottesman, Daniel},
  year={1997},
DOI={10.48550/arXiv.quant-ph/9705052},
  publisher={California Institute of Technology}
}

@article{gottesman1996class,
  title={Class of quantum error-correcting codes saturating the quantum Hamming bound},
  author={Gottesman, Daniel},
  journal={Physical Review A},
  volume={54},
  number={3},
  pages={1862},
  year={1996},
DOI={10.1103/PhysRevA.54.1862},
  publisher={APS}
}

@book{lidar2013quantum,
  title={Quantum error correction},
  author={Lidar, Daniel A and Brun, Todd A},
  year={2013},
DOI={10.1017/cbo9781139034807},
  publisher={Cambridge university press}
}

@article{conrad2022gottesman,
  title={Gottesman-Kitaev-Preskill codes: A lattice perspective},
  author={Conrad, Jonathan and Eisert, Jens and Arzani, Francesco},
  journal={Quantum},
  volume={6},
  pages={648},
  year={2022},
DOI={10.22331/q-2022-02-10-648},
  publisher={Verein zur F{\"o}rderung des Open Access Publizierens in den Quantenwissenschaften}
}

@article{sarkar2023qudit,
  title={The qudit Pauli group: non-commuting pairs, non-commuting sets, and structure theorems},
  author={Sarkar, Rahul and Yoder, Theodore J},
  journal={arXiv preprint arXiv:2302.07966},
DOI={10.48550/arXiv.2302.07966},
  year={2023}
}

@article{vuillot2024homological,
  title={Homological quantum rotor codes: Logical qubits from torsion},
  author={Vuillot, Christophe and Ciani, Alessandro and Terhal, Barbara M},
  journal={Communications in Mathematical Physics},
  volume={405},
  number={2},
  pages={53},
  year={2024},
  publisher={Springer}
}

@article{conrad2023good,
  title={Good Gottesman-Kitaev-Preskill codes from the NTRU cryptosystem},
  author={Conrad, Jonathan and Eisert, Jens and Seifert, Jean-Pierre},
  journal={arXiv preprint arXiv:2303.02432},
DOI={10.48550/arXiv.2303.02432},
  year={2023}
}

@article{gunderman2024stabilizer,
  title={Stabilizer codes with exotic local-dimensions},
  author={Gunderman, Lane G},
  journal={Quantum},
  volume={8},
  pages={1249},
  year={2024},
  publisher={Verein zur F{\"o}rderung des Open Access Publizierens in den Quantenwissenschaften}
}

@article{michael2016new,
  title={New class of quantum error-correcting codes for a bosonic mode},
  author={Michael, Marios H and Silveri, Matti and Brierley, RT and Albert, Victor V and Salmilehto, Juha and Jiang, Liang and Girvin, Steven M},
  journal={Physical Review X},
  volume={6},
  number={3},
  pages={031006},
  year={2016},
  publisher={APS}
}

@article{novak2024homological,
  title={Homological Quantum Error Correction with Torsion},
  author={Nov{\'a}k, Samo},
  journal={arXiv preprint arXiv:2405.03559},
  year={2024}
}

@article{chakraborty2025hybrid,
  title={Hybrid oscillator-qudit quantum processors: stabilizer states and symplectic operations},
  author={Chakraborty, Sayan and Albert, Victor V},
  journal={arXiv preprint arXiv:2508.04819},
  year={2025}
}

@article{gunderman2025beyond,
  title={Beyond integral-domain stabilizer codes},
  author={Gunderman, Lane G},
  journal={Physical Review A},
  volume={112},
  number={6},
  pages={062430},
  year={2025},
  publisher={APS}
}

@article{chuang1997bosonic,
  title={Bosonic quantum codes for amplitude damping},
  author={Chuang, Isaac L and Leung, Debbie W and Yamamoto, Yoshihisa},
  journal={Physical Review A},
  volume={56},
  number={2},
  pages={1114},
  year={1997},
  publisher={APS}
}

@article{albert2018performance,
  title={Performance and structure of single-mode bosonic codes},
  author={Albert, Victor V and Noh, Kyungjoo and Duivenvoorden, Kasper and Young, Dylan J and Brierley, RT and Reinhold, Philip and Vuillot, Christophe and Li, Linshu and Shen, Chao and Girvin, Steven M and others},
  journal={Physical Review A},
  volume={97},
  number={3},
  pages={032346},
  year={2018},
  publisher={APS}
}

@inproceedings{li2025qutrit,
  title = {A Qutrit $[\![6,2,2]\!]_3$ Quantum Error-Correcting Code with a Transversal Binary AND Gate},
  author = {Li, Christine and Yeh, Lia},
  booktitle = {2025 IEEE International Conference on Quantum Computing and Engineering (QCE)},
  volume = {2},
  pages = {550--551},
  year = {2025},
  organization = {IEEE}
}

@article{ball2026error,
  title={Error-correcting codes and absolutely maximally entangled states for mixed-dimensional Hilbert spaces},
  author={Ball, Simeon and Zhang, Raven},
  journal={Physical Review A},
  volume={113},
  number={1},
  pages={012432},
  year={2026},
  publisher={APS}
}

@article{li2026transversal,
  title={Transversal AND in Quantum Codes},
  author={Li, Christine and Yeh, Lia},
  journal={arXiv preprint arXiv:2603.04548},
  year={2026}
}

@article{wilde2008optimal,
  title={Optimal entanglement formulas for entanglement-assisted quantum coding},
  author={Wilde, Mark M and Brun, Todd A},
  journal={Physical Review A—Atomic, Molecular, and Optical Physics},
  volume={77},
  number={6},
  pages={064302},
  year={2008},
  publisher={APS}
}

@article{gunderman2023transforming,
  title={Transforming collections of Pauli operators into equivalent collections of Pauli operators over minimal registers},
  author={Gunderman, Lane G},
  journal={Physical Review A},
  volume={107},
  number={6},
  pages={062416},
  year={2023},
  publisher={APS}
}

@article{gunderman2024minimal,
  title={Minimal qubit representations of Hamiltonians via conserved charges},
  author={Gunderman, Lane G and Jena, Andrew and Dellantonio, Luca},
  journal={Physical Review A},
  volume={109},
  number={2},
  pages={022618},
  year={2024},
  publisher={APS}
}

@article{huber2025quasi,
  title={Quasi-Clifford to qubit mappings},
  author={Huber, Felix},
  journal={arXiv preprint arXiv:2508.01470},
  year={2025}
}

@article{brenner2025factoring,
  title={Factoring an integer with three oscillators and a qubit},
  author={Brenner, Lukas and Caha, Libor and Coiteux-Roy, Xavier and Koenig, Robert},
  journal={Nature Communications},
  year={2025},
  publisher={Nature Publishing Group UK London}
}

@article{gross2006hudson,
  title={Hudson’s theorem for finite-dimensional quantum systems},
  author={Gross, David},
  journal={Journal of mathematical physics},
  volume={47},
  number={12},
  year={2006},
  publisher={AIP Publishing}
}

@article{howard2014contextuality,
  title={Contextuality supplies the ‘magic’for quantum computation},
  author={Howard, Mark and Wallman, Joel and Veitch, Victor and Emerson, Joseph},
  journal={Nature},
  volume={510},
  number={7505},
  pages={351--355},
  year={2014},
  publisher={Nature Publishing Group UK London}
}

@article{eastin2009restrictions,
  title={Restrictions on transversal encoded quantum gate sets},
  author={Eastin, Bryan and Knill, Emanuel},
  journal={Physical review letters},
  volume={102},
  number={11},
  pages={110502},
  year={2009},
  publisher={APS}
}

@article{zeng2011transversality,
  title={Transversality versus universality for additive quantum codes},
  author={Zeng, Bei and Cross, Andrew and Chuang, Isaac L},
  journal={IEEE Transactions on Information Theory},
  volume={57},
  number={9},
  pages={6272--6284},
  year={2011},
  publisher={IEEE}
}

\clearpage

\if{false}
\begin{lemma}
Let $P$ be a set of compositionally independent Pauli operators for a device $D$ with $\mathsf{Regdim}_{1\Rightarrow u}(D)=[q_1,q_2,\ldots, q_u]$, then these can be extended into generators for a mixed-register stabilizer code with $n=u+\frac{1}{2}\mathsf{rank}([\tilde{\odot}_u])$ total registers, where $[\tilde{\odot}_u]$ is the matrix of generalized symplectic inner product with $[\tilde{\odot}_u]_{ij}=\phi(p_i)\tilde{\odot}_u\phi(p_j)$ for $p_i,p_j\in P$.  
\end{lemma}

\fi


\if{false}
\begin{proof}
First, note that under composition of Pauli operator $i$ onto operator $j$, the action on the matrix $[\tilde{\odot}_u]$ is to simultaneously add row and column $i$ to row and column $j$. As $[\tilde{\odot}_u]$ must be hollow, this simultaneous action is equivalent to a row addition followed by a column addition (or conjugation by a row addition matrix). The order of each element in $P$ must be at most $\mathsf{lcm}(q_1,q_2,\ldots,q_u)$. Through repeated compositions, we obtain all the primitives required for Gaussian elimination over the module. Notably, we may reduce $[\tilde{\odot}_u]$ into a pseudodiagonal form: $\oplus_{i=1}^{\mathsf{rank}([\tilde{\odot}_u])/2} \begin{bmatrix}0&-1\\ 1&0\end{bmatrix}\oplus_{i=1}^{|P|-\mathsf{rank}([\tilde{\odot}_u])/2}[0]$. Let $L$ be the compositions performed, then we may complete the commutations by appending $\{X_{\mathsf{lcm}},Z_{\mathsf{lcm}}\}$ to the end of each operator in the first $\mathsf{rank}([\tilde{\odot}_u])/2$ positions, then conjugate these new registers by $L^{-1}$ to determine the appropriate appended operators so that the full set is a valid mixed-register stabilizer code. [This argument doesn't quite work, instead use Sarkar+Yoder's result and then argue that the order must be factors of the LCM]

\end{proof}
\fi


\if{false}
\begin{theorem}
Let $S$ be a stabilizer code with parameters $[[n,k,d]]_q$. Then there exists a mixed-register device $D$, with at least two different register dimensions, such that there is a selection of $\{q_i\}$ for $\mathsf{Regdim}(D)$ such that $S$ can be transformed into an $[[n,k',d']]_{\vec{q}}$ stabilizer code, with $d'\geq d$ and $k'\geq k$. [$k'\geq k$ since the order of the operators could diminish if the numerator happens to cancel out nicely]
\end{theorem}

\begin{proof}
For $k'$, the number of independent generators is preserved, however, the order of the generators may not be. It's possible that some of the generators have order only a factor of $\mathsf{lcm}(q_i)$, thus increasing the logical space. The missing piece is showing that there exists a partition always (although a really uneven one should work).

For $d'$ this will use the argument from "Beyond integral-domain stabilizer codes" by Gunderman---which means that the local-dimensions of the registers may need to be comically large but this is just an existence.
\end{proof}

\begin{remark}
Note that the above result immediately provides good qLDPC mixed-register codes, albeit with caveats on the selection of register dimensions. A worthwhile endeavor would be to show the reverse case: given a selection of register dimensions there exists a stabilizer code from a uniform finite dimension such that the parameters are at least preserved.
\end{remark}
\fi



\section*{Appendix}

Here we provide a proof of our structural decomposition Theorem.

\begin{proof}[Proof of Theorem \ref{thm:decomposition}]


Define the (normal) subgroup
\begin{gather}
\mathcal{R}(\mathcal{T})=\{W\in\mathcal{T}\mid\phi(W)\tilde{\odot}\phi(P)=0\ \forall P\in\mathcal{T}\}.\label{eq:radical_def}
\end{gather}
Choose a minimal generating set $\{W_1,\dots,W_\ell\}$ for $\mathcal{R}(\mathcal{T})$. Hence,
\begin{gather}
\mathcal{R}(\mathcal{T})=\Left\langle W_1,W_2,\dots,W_\ell\Right\rangle.\label{eq:radical_gen}
\end{gather}
Let $\mathcal{K}=\mathcal{T}/\mathcal{R}(\mathcal{T})$. If $\mathcal{K}$ is the trivial group, then $\mathcal{R}(\mathcal{T})=\mathcal{T}\Rightarrow\mathcal{T}$ is a mixed-register stabilizer code and we are done. Assume $\mathcal{K}$ is nontrivial. Define the alternating bicharacter $\beta:\mathcal{K}\times\mathcal{K}\to\frac{1}{Q}\mathbb{Z}/\mathbb{Z}$, where $\frac{1}{Q}\mathbb{Z}/\mathbb{Z}<\mathbb{Q}/\mathbb{Z}$, as
\begin{gather}
\beta(P\mathcal{R}(\mathcal{T}),Q\mathcal{R}(\mathcal{T}))=\phi(P)\tilde{\odot}\phi(Q),\text{ where }P,Q\in\mathcal{T}.
\end{gather}
It is necessary to see that this is indeed well-defined since picking any other representative for the coset $P\mathcal{R}(\mathcal{T})$, say $P^\prime=PW$ for some $W\in \mathcal{R}(\mathcal{T})$, does not change $\beta$:
\begin{eqnarray}
\beta(P^\prime \mathcal{R}(\mathcal{T}), Q\mathcal{R}(\mathcal{T})) &=& \phi(P^\prime)\tilde{\odot}\phi(Q)\\
&=& \phi(PW)\tilde{\odot}\phi(Q)\\
&=& (\phi(P)+\phi(W))\tilde{\odot}\phi(Q)\\
&=& \phi(P)\tilde{\odot}\phi(Q) + \underbrace{\phi(W)\tilde{\odot}\phi(Q)}_{0\text{ by }\eqref{eq:radical_def}}\\
&=& \beta(P\mathcal{R}(\mathcal{T}),Q\mathcal{R}(\mathcal{T})),
\end{eqnarray}
and similarly for replacing $Q$. By way of contradiction, assume that there exists a non-identity element $P\mathcal{R}(\mathcal{T})\in\mathcal{K}$ that pairs trivially under $\beta$ with all other elements; i.e., $\beta(P\mathcal{R}(\mathcal{T}),Q\mathcal{R}(\mathcal{T}))=0$ for all $Q\mathcal{R}(\mathcal{T})\in\mathcal{K}$. Then,
\begin{eqnarray}
&&\phi(P)\tilde{\odot}\phi(Q)=0\text{ for all }Q\in\mathcal{T}\\
&\Rightarrow& P\in \mathcal{R}(\mathcal{T})\text{ by the definition of }\mathcal{R}(\mathcal{T})\\
&\Rightarrow& P\mathcal{R}(\mathcal{T})\text{ is the identity coset in }\mathcal{K}\\
&\Rightarrow& \text{Contradiction (since we assumed }P\mathcal{R}(\mathcal{T})\text{ is not the identity)}.
\end{eqnarray}
Hence, supposing that $I$ is the identity Pauli,
\begin{gather}
\forall P\mathcal{R}(\mathcal{T})\in\mathcal{K}\text{ with }P\mathcal{R}(\mathcal{T})\ne I\mathcal{R}(\mathcal{T}),\ \exists Q\mathcal{R}(\mathcal{T})\in\mathcal{K}\text{ such that }\beta(P\mathcal{R}(\mathcal{T}),Q\mathcal{R}(\mathcal{T}))\ne0.\label{eq:nontrivial_pairing} 
\end{gather}
Choose $U\in\mathcal{K}$ of maximal order, say $\lvert U\rvert=d$. Then, $d\mid Q$ because $U^Q=I\mathcal{R}(\mathcal{T})$. Define the mapping
\begin{gather}
\beta_{U}:\mathcal{K}\to\frac{1}{Q}\mathbb{Z}/\mathbb{Z}\text{ given by }A\mapsto \beta(U,A).
\end{gather}
Then, the map $\beta_{U}$ is a group homomorphism because the pairing $\beta$ is additive in the second argument. For every $A\in\mathcal{K}$,
\begin{gather}
d\beta_{U}(A)=d\beta(U,A)=\beta(U^d,A)=\beta(I\mathcal{R}(\mathcal{T}),A)=0.
\end{gather}
Hence, every element of the image $\beta_{U}(\mathcal{K})$ has an order that divides $d$. Moreover, $\beta_{U}(\mathcal{K})$ is cyclic because it is a subgroup of the cyclic group $\frac{1}{Q}\mathbb{Z}/\mathbb{Z}$. Thus, $\lvert\beta_{U}(\mathcal{K})\rvert\text{ divides }d$. Suppose, by way of contradiction, that $\lvert\beta_{U}(\mathcal{K})\rvert=d^\prime<d$. Then,
\begin{eqnarray}
&&d^\prime\beta_{U}(A)=0\text{ for all }A\in\mathcal{K}\\
&\Rightarrow& d^\prime\beta(U,A)=0\text{ for all }A\in\mathcal{K}\\
&\Rightarrow& \beta\Left(U^{d^\prime},A\Right)=0\text{ for all }A\in\mathcal{K}.\\
&\Rightarrow& U^{d^\prime}=I\mathcal{R}(\mathcal{T})\text{ by }\eqref{eq:nontrivial_pairing}\\
&\Rightarrow& \text{Contradiction (since }\lvert U\rvert=d>d^\prime\text{)}.
\end{eqnarray}


\noindent Thus, $\lvert\beta_{U}(\mathcal{K})\rvert=d$. By the fundamental theorem of cyclic groups, $\beta_{U}(\mathcal{K})$ is the unique order-$d$ (cyclic) subgroup of $\frac{1}{Q}\mathbb{Z}/\mathbb{Z}$, and since $\frac{1}{d}\in\frac{1}{Q}\mathbb{Z}/\mathbb{Z}$ is an element of order $d$, it generates $\beta_{U}(\mathcal{K})$:
\begin{gather}
\beta_{U}(\mathcal{K})=\Left\langle\frac{1}{d}\Right\rangle=\Left\{0,\frac{1}{d},\dots,\frac{d-1}{d}\Right\}<\frac{1}{Q}\mathbb{Z}/\mathbb{Z}.
\end{gather}
Since $\frac{1}{d}\in\beta_{U}(\mathcal{K})$, there exists an element $V\in\mathcal{K}$ such that $\beta_{U}(V)=\frac{1}{d}$, that is $\beta(U,V)=\frac{1}{d}$. Then, by properties of group homomorphisms, $\lvert\beta_U(V)\rvert$ divides $\lvert V\rvert$; i.e., $d\mid\lvert V\rvert$. This implies $d\leq\lvert V\rvert$. But $U$ was chosen of maximal order, so no element of $\mathcal{K}$ has order exceeding $d$. Thus, $\lvert V\rvert=d$. Let $\mathcal{G}$ be the subgroup of $\mathcal{K}$ generated by $U$ and $V$; i.e., $\mathcal{G}=\langle U,V\rangle\leq\mathcal{K}$. The next claim is that $\mathcal{G}$ is precisely the internal direct product $\langle U\rangle\times\langle V\rangle$. To show this, it suffices to show that the intersection $\langle U\rangle\cap\langle V\rangle$ is trivial. By way of contradiction, assume that there exists a non-identity element $A\in\langle U\rangle\cap\langle V\rangle$. Then, there exist $a,b\in\mathbb{Z}$ such that $U^a=V^b=A$. Hence,
\begin{eqnarray}
&&\beta_U(V^b)=\beta_U(U^a)\\
&\Rightarrow& b\beta_U(V)=0\\
&\Rightarrow& \frac{b}{d}=0\mod 1\\
&\Rightarrow& d\mid b\\
&\Rightarrow& A=V^b=(V^d)^{\frac{b}{d}}=I\mathcal{R}(\mathcal{T})\\
&\Rightarrow& \text{Contradiction (since we assumed }A\text{ to be non-identity)}.
\end{eqnarray}
Hence, the intersection $\langle U\rangle\cap\langle V\rangle$ is in fact trivial, which implies that $\mathcal{G}=\langle U\rangle\times\langle V\rangle$. But both $\langle U\rangle$ and $\langle V\rangle$ are cyclic groups of order $d$. Thus,
\begin{gather}
\mathcal{G}\cong\mathbb{Z}_d\oplus\mathbb{Z}_d.
\end{gather}
Define
\begin{gather}
\mathcal{G}^\perp=\{A\in\mathcal{K}\mid\beta(A,B)=0\ \forall B\in\mathcal{G}\}.
\end{gather}
Since $\mathcal{G}$ is generated by $U$ and $V$, this is equivalently
\begin{gather}
\mathcal{G}^\perp=\{A\in\mathcal{K}\mid\beta(A,U)=\beta(A,V)=0\}.
\end{gather}
Note that $\mathcal{G}\cap\mathcal{G}^\perp=\{I\mathcal{R}(\mathcal{T})\}$. Suppose, by way of contradiction, that this is not true and that there exists some non-identity element $A\in\mathcal{G}\cap\mathcal{G}^\perp$. Then,
\begin{eqnarray}
&\because& A\in\mathcal{G},\ A=U^aV^b\text{ for some }a,b\in\mathbb{Z},\\
&\because& A\in\mathcal{G}^\perp,\ \beta(A,U)=\beta(A,V)=0\\
&\Rightarrow& \beta(U^aV^b,U)=\beta(U^aV^b,V)=0\\
&\Rightarrow& b\beta(V,U)=a\beta(U,V)=0\\
&\Rightarrow& d\mid b\text{ and }d\mid a\\
&\Rightarrow& U^a=V^b=I\mathcal{R}(\mathcal{T})\\
&\Rightarrow& A=U^aV^b=I\mathcal{R}(\mathcal{T})\\
&\Rightarrow& \text{Contradiction (since we assumed }A\text{ to be non-identity)}.
\end{eqnarray}
Hence, $\mathcal{G}^\perp$ is the orthogonal complement of $\mathcal{G}$ in $\mathcal{K}$. For any $A\in\mathcal{K}$, let
\begin{gather}
a_A=d\beta(A,V)\in\mathbb{Z}\quad\text{and}\quad b_A=-d\beta(A,U)\in\mathbb{Z}.
\end{gather}
Then set
\begin{gather}
A^\prime=A\Left(U^{a_A}V^{b_A}\Right)^{-1}.
\end{gather}
This implies that
\begin{eqnarray}
\beta(A^\prime,U)&=\beta(A,U)-b_A\beta(V,U) = \beta(A,U)-\frac{d\beta(A,U)}{d}=0,\text{ and}\\
\beta(A^\prime,V)&=\beta(A,V)-a_A\beta(U,V) = \beta(A,V)-\frac{d\beta(A,V)}{d}=0.
\end{eqnarray}
Hence, $A^\prime\in\mathcal{G}^\perp$. Thus, every $A\in\mathcal{K}$ can be written as
\begin{gather}
A=\Left(U^{a_A}V^{b_A}\Right)A^\prime,
\end{gather}
where $U^{a_A}V^{b_A}\in\mathcal{G}$ and $A^\prime\in\mathcal{G}^\perp$. Equivalently,
\begin{gather}
\text{for all }A\in\mathcal{K},\text{ there exist }\tilde{A}\in\mathcal{G}\text{ and }A^\prime\in\mathcal{G}^\perp\text{ such that }A=\tilde{A}A^\prime.\label{eq:idp_cond1}
\end{gather}
From \eqref{eq:idp_cond1} and the fact that the intersection $\mathcal{G}\cap\mathcal{G}^\perp$ is trivial, we conclude that $\mathcal{K}$ can, in fact, be written as the internal direct product between $\mathcal{G}$ and $\mathcal{G}^\perp$:
\begin{gather}
\mathcal{K}=\mathcal{G}\times\mathcal{G}^\perp.
\end{gather}
This is one mixed-register symplectic Gram-Schmidt step. We subscript the variables used herein with a "1" to explicitly denote the first step, and we add a tilde to the notation whose purpose will be clear later in the proof:
\begin{itemize}
    \item We chose a $\tilde{U}_1\in\mathcal{K}$ of maximal order $d_1$ and showed that there exists a $\tilde{V}_1\in\mathcal{K}$ also of order $d_1$ such that $\beta\Left(\tilde{U}_1,\tilde{V}_1\Right)=\frac{1}{d_1}$.
    \item We defined $\mathcal{G}_1=\Left\langle\tilde{U}_1,\tilde{V}_1\Right\rangle$ and showed that $\mathcal{G}_1=\Left\langle\tilde{U}_1\Right\rangle\times\Left\langle\tilde{V}_1\Right\rangle\cong\mathbb{Z}_{d_1}\oplus\mathbb{Z}_{d_1}$.
    \item We defined $\mathcal{G}_1^\perp$ as the orthogonal complement of $\mathcal{G}_1$ in $\mathcal{K}$ and showed that $\mathcal{K}=\mathcal{G}_1\times\mathcal{G}_1^\perp\cong(\mathbb{Z}_{d_1}\oplus\mathbb{Z}_{d_1})\oplus\mathcal{G}_1^\perp$.
\end{itemize}
If $\mathcal{G}_1^\perp$ is trivial, then we have our desired decomposition. This is because we have $\mathcal{K}=\mathcal{G}_1=\Left\langle\tilde{U}_1,\tilde{V}_1\Right\rangle$; i.e.,
\begin{gather}
\mathcal{T}/\mathcal{R}(\mathcal{T})=\Left\langle\tilde{U}_1,\tilde{V}_1\Right\rangle.\label{eq:K_decomp}
\end{gather}
Now lift these generators back to $\mathcal{T}$: Let $U_1,V_1\in\mathcal{T}$ be coset representatives for $\tilde{U}_1,\tilde{V}_1\in\mathcal{K}$ respectively. Then, from \eqref{eq:radical_def}, \eqref{eq:radical_gen}, and \eqref{eq:K_decomp}, we find that $\mathcal{T}$ is generated by $W_1,W_2,\dots,W_\ell,U_1,V_1$; i.e.,
\begin{gather}
\mathcal{T}=\Left\langle W_1,W_2,\dots,W_\ell,U_1,V_1\Right\rangle,
\end{gather}
and we are done. If $\mathcal{G}_1^\perp$ is nontrivial, then we recursively perform the mixed-register symplectic Gram-Schmidt step on $\mathcal{G}_1^\perp$ and express it as
\begin{eqnarray}
\mathcal{G}_1^\perp&=&\mathcal{G}_2\times\mathcal{G}_2^\perp\\
&=&\Left\langle\tilde{U_2},\tilde{V}_2\Right\rangle\times\mathcal{G}_2^\perp\ \Left(\text{for some }\tilde{U}_2,\tilde{V}_2\in\mathcal{G}_1^\perp\text{ of maximal order }d_2\Right)\\
&=&\Left\langle\tilde{U}_2\Right\rangle\times\Left\langle\tilde{V}_2\Right\rangle\times\mathcal{G}_2^\perp\\
&\cong&\mathbb{Z}_{d_2}\oplus\mathbb{Z}_{d_2}\oplus\mathcal{G}_2^\perp.
\end{eqnarray}
We terminate this recursive mixed-register symplectic Gram-Schmidt procedure (and get our desired decomposition) after $c$ steps if $\mathcal{G}_c^\perp$ is trivial. Also, $c$ is necessarily finite because if we define $\mathcal{G}_0^\perp=\mathcal{K}$, then for all $i\in\{1,\dots,c\}$, $\Left\lvert\mathcal{G}_i^\perp\Right\rvert<\Left\lvert\mathcal{G}_{i-1}^\perp\Right\rvert$. More precisely,
\begin{gather}
\Left\lvert\mathcal{G}_i^\perp\Right\rvert=\frac{\Left\lvert\mathcal{G}_{i-1}^\perp\Right\rvert}{d_i^2}.
\end{gather}
On termination after $c$ steps, we have the decomposition
\begin{eqnarray}
\mathcal{T}/\mathcal{R}(\mathcal{T})&=&\mathcal{G}_1\times\mathcal{G}_2\times\dots\times\mathcal{G}_c\\
&=&\prod_{i=1}^c\Left\langle\tilde{U}_i,\tilde{V}_i\Right\rangle\ \left(\text{where }\tilde{U}_i,\tilde{V}_i\in\mathcal{K}\text{ with }\Left\lvert\tilde{U}_i\Right\rvert=\Left\lvert\tilde{V}_i\Right\rvert=d_i\right)\label{eq:final_K_UV_decomp}\\
&=&\prod_{i=1}^c\Left(\Left\langle\tilde{U}_i\Right\rangle\times\Left\langle\tilde{V}_i\Right\rangle\Right)\\
&\cong&\bigoplus_{i=1}^c\Left(\mathbb{Z}_{d_i}\oplus\mathbb{Z}_{d_i}\Right).\label{eq:final_K_cylic_decomp}
\end{eqnarray}
Therefore,
\begin{gather}
\boxed{\mathcal{T}\cong\mathcal{R}(\mathcal{T})\oplus\bigoplus_{i=1}^c\Left(\mathbb{Z}_{d_i}\oplus\mathbb{Z}_{d_i}\Right)}.\label{eq:main_decomp}
\end{gather}
Let $U_1,V_1,\dots,U_c,V_c\in\mathcal{T}$ be coset representatives for $\tilde{U}_1,\tilde{V}_1,\dots,\tilde{U}_c\tilde{V}_c\in\mathcal{K}$ respectively. Then, from \eqref{eq:radical_def}, \eqref{eq:radical_gen}, and \eqref{eq:final_K_UV_decomp}, we can write
\begin{gather}
\mathcal{T}=\langle W_1,W_2,\dots,W_\ell,U_1,V_1,U_2,V_2,\dots,U_c,V_c\rangle
\end{gather}
such that
\begin{enumerate}
    \item the generators satisfy the following relations (all mod 1):
    \begin{enumerate}
        \item $\phi(W_i)\tilde{\odot}\phi(P)=0$ for all $P\in\mathcal{T}$, for all $i\in\{1,\dots,\ell\}$
        \begin{quote}
            (from \eqref{eq:radical_def} and \eqref{eq:radical_gen}),
        \end{quote}
        \item $\phi(U_i)\tilde{\odot}\phi(U_j)=\phi(V_i)\tilde{\odot}\phi(V_j)=0$ for all $i,j\in\{1,\dots,c\}$
        \begin{quote}
            ($\because$ $U_i\in\mathcal{G}_j^\perp$ for all $j<i$ and $U_j\in\mathcal{G}_i^\perp$ for all $j>i$, and similarly for the $V$'s),
        \end{quote}
        \item $\phi(U_i)\tilde{\odot}\phi(V_j)=0$ for all $i\ne j$
        \begin{quote}
            ($\because$ $U_i\in\mathcal{G}_j^\perp$ for all $j<i$ and $V_j\in\mathcal{G}_i^\perp$ for all $j>i$),
        \end{quote}
        \item $\phi(U_i)\tilde{\odot}\phi(V_i)=\frac{1}{d_i} (\!\!\!\!\mod 1)$ with $d_i\mid Q$
        \begin{quote}
            ($\because$ $\lvert U_i\rvert=\lvert V_i\rvert=d_i\Rightarrow\beta(U\mathcal{R}(\mathcal{T})_i,V_i\mathcal{R}(\mathcal{T}))=\frac{1}{d_i}$),
        \end{quote}
    \end{enumerate}
    \item $d_1\mid d_2\mid\ldots\mid d_c$
    \begin{quote}
        (by the fundamental theorem of finite Abelian groups, the cyclic groups in the decomposition given in \eqref{eq:final_K_cylic_decomp} can be (re)written in invariant-factor form; i.e., they can be reordered such that this divisibility chain condition for their orders is satisfied),
    \end{quote}
    \item the multiset $\{d_1,\ldots.d_c\}$ is unique depending only on $\mathcal{T}$
    \begin{quote}
        ($\because$ the fundamental theorem of finite Abelian groups guarantees that the decomposition of $\mathcal{K}$ into cyclic groups given in $\eqref{eq:final_K_cylic_decomp}$ is unique and depends only on $\mathcal{K}$ and the multiset in question simply contains the orders of those cyclic groups and $\mathcal{K}$ depends only on $\mathcal{T}$).
    \end{quote}
\end{enumerate}
\end{proof}

\end{document}